\documentclass[a4paper,onecolumn,11pt]{quantumarticle}
\pdfoutput=1
\usepackage[utf8]{inputenc}
\usepackage[english]{babel}
\usepackage[T1]{fontenc}
\usepackage{amsmath}

\usepackage{tikz}
\usepackage{lipsum}

\usepackage[utf8]{inputenc}
\usepackage[T1]{fontenc}
\usepackage{amsmath,amsxtra,amssymb,amsthm,amsfonts}

\usepackage{xfrac}
\usepackage{mathrsfs}
\usepackage{mathtools}   
\usepackage{tikz}
\usepackage{xcolor}
\usetikzlibrary{arrows, automata}

\usepackage{dsfont}

\usepackage{fullpage}
\usepackage{enumerate}
\usepackage[colorlinks=true,linkcolor=blue,citecolor=blue,urlcolor=blue]{hyperref}

\definecolor{dblue}{rgb}{0, 0, 0.72}

\numberwithin{equation}{section}

\theoremstyle{definition}
\newtheorem{question*}{Question}

\newcommand{\C}{{\mathbb C}}
\newcommand{\ten}{\otimes}

\newcommand{\pl}{\hspace{.1cm}}

\newcommand{\al}{\alpha}
\newcommand{\si}{\sigma}

\newcommand{\la}{\lambda}
\newcommand{\eps}{\varepsilon}

\newcommand{\id}{\iota_{\infty,2}^n}

\newcommand{\ha}{{\mathcal H}}

\newcommand{\F}{{\mathcal F}}
\newcommand{\E}{{\mathcal E}}

\newcommand{\A}{{\mathcal A}}
\newcommand{\B}{{\mathcal{B}}}

\newcommand{\M}{{\mathcal M}}
\newcommand{\cH}{{\mathcal H}}
\newcommand{\cK}{{\mathcal K}}
\newcommand{\cF}{{\mathcal F}}
\newcommand{\cS}{{\mathcal S}}

\newcommand{\K}{{\mathcal K}}
\renewcommand{\S}{{\mathcal S}}

\newcommand{\Ha}{{\mathcal H}}

\newcommand{\tr}{\operatorname{tr}}

\newcommand{\N}{{\mathcal N}}

\newcommand{\I}{{\mathcal I}}

\newcommand{\norm}[2]{\parallel \! #1 \! \parallel_{#2}}

\newcommand{\bra}[1]{\langle{#1}|}
\newcommand{\ket}[1]{|{#1}\rangle}
\newcommand{\ketbra}[1]{|{#1}\rangle\langle{#1}|}
\newcommand{\qd}{\end{proof}\vspace{0.5ex}}
\newcommand{\prf}{\begin{proof}[\bf Proof:]}

\newcommand{\xspace}{\hbox{\kern-2.5pt}}

\renewcommand{\id}{\operatorname{id}}
\renewcommand{\eps}{\varepsilon}
\renewcommand{\al}{\alpha}
\renewcommand{\E}{{\mathcal{E}}}

\newcommand{\Ka}{{\mathcal K}}

\def\ket#1{| #1 \rangle}
\def\bra#1{\langle #1 |}

\theoremstyle{plain}
\newtheorem{theorem}{Theorem}
\newtheorem{lemma}[theorem]{Lemma}
\newtheorem{corollary}[theorem]{Corollary}
\newtheorem{prop}[theorem]{Proposition}

\newtheorem*{conjecture*}{Conjecture}
\newtheorem{remark}[theorem]{Remark}

\theoremstyle{definition}
\newtheorem{definition}[theorem]{Definition}

\newtheorem{exam}[theorem]{Example}

\usepackage{hyperref}

\begin{document}

\title{Generalized Stein's lemma and asymptotic equipartition property for subalgebra entropies}

\author{Li Gao}
\affiliation{School of Mathematics and Statistics, Wuhan University, Wuhan, Hubei 430075, China\\ Wuhan Institute of Quantum Technology, 
Wuhan, Hubei 430075, China}
 \email{gao.li@whu.edu.cn}
 \author{Mizanur Rahaman}
 \affiliation{Univ Lyon, Inria, ENS Lyon, UCBL, LIP, F-69342, Lyon Cedex 07, France\\
 Wallenberg Centre for Quantum Technology, Chalmers University of Technology, \\
Department of Mathematical Sciences, Chalmers University of Technology}
 \email{mizanurr@chalmers.se}


\begin{abstract}The quantum Stein’s lemma is a fundamental result of quantum hypothesis testing in the context of distinguishing two quantum states. The ``generalized quantum Stein's lemma'' asserts that this result is true in a general framework where one of the states is replaced by convex sets of quantum states. Formulated in 2008, the generalized Stein's lemma is one of the most influential results in quantum information theory that links resource convertibility to quantum hypothesis testing. However, in 2023 a logical gap was found in the original proof of the generalized Stein's lemma and since then there has been an enormous effort in resolving this issue. 
In this work, we show that the assertion of the generalized Stein's lemma is true for the setting where the second hypothesis is the state space of any finite-dimensional subalgebra $\N$. This is obtained through a strong asymptotic equipartition property for smooth subalgebra entropies that applies to any fixed smoothing parameter $\eps\in (0,1)$. In fact, we obtain a stronger second-order analysis for the hypothesis-testing relative entropy in this setting. As an application in resource theory, we show that the relative entropy of a subalgebra is the asymptotic dilution cost under suitable operations. This provides a possible route to establishing a connection between different quantum resources based on subalgebras. After finishing this article, we learned that there are two independent works (\cite{Hayashi-yamasaki, Lami-gqsl}) that finally resolve the generalized Stein's lemma in its full generality. However, we provide an alternative proof in a special case using different operator-algebraic techniques that may be of independent interest.


\end{abstract}
\maketitle
\section{Introduction}

Hypothesis testing is one of the most fundamental tasks in information theory. It involves making decisions based on experimental data from random variables, using tools and techniques from statistics and probability theory. In the setting of quantum information, and in particular in the context of quantum state discrimination, hypothesis testing is deeply connected with many important tasks of quantum information processing (\cite{HN, Loy, brandao-plenio}).

When identifying whether a given quantum state is one of the two possibilities, state $\rho$ or state $\sigma$, one can make two types of errors: the observer finds that the state is $\sigma$, when in reality it was $\rho$ (type $\mathrm{I}$ error), and the other type of error where the observer finds the state to be $\rho$, when it was  actually $\sigma$ (type $\mathrm{II}$ error). In the setting of \textbf{asymmetric hypothesis testing} one defines the optimized type $\mathrm{II}$ error $\beta_{\eps}(\rho\|\sigma)$ as the minimum probability of type $\mathrm{II}$ error, while requiring that the probability of type $\mathrm{I}$  error is bounded by a small parameter $\eps$. The celebrated \textbf{quantum Stein's lemma}  (\cite{HP, ON} )
asserts that when the source emits i.i.d. copies of one of two quantum states $\rho$ and $\sigma$, then for any $0<\eps <1$,
\[\beta_{\eps}(\rho^{\otimes n}\|\sigma^{\otimes n}) \sim 2^{-n D(\rho\|\sigma)}, \  \  \  \text{as} \ \ n\to \infty, \]
where $D(\rho\|\sigma)=\tr(\rho(\log \rho-\log \sigma))$ is the quantum relative entropy of $\rho$ and $\sigma$.

In a more intricate situation known as composite hypothesis testing (\cite{BBH, BG, brandao-plenio}), instead of two quantum states $\rho$ and $\sigma$,  one wishes to distinguish between a state $\rho$ and a whole family of quantum states. For example, to test whether a quantum device produces entanglement, one would like to distinguish a bipartite state against all the separable states. In the context of composite hypothesis testing, when the family of quantum states fulfills certain axioms motivated by the framework of resource theories, Brand\~{a}o and Plenio (\cite{brandao-plenio}) claimed to prove a \textbf{generalized Stein's lemma} showing that the asymptotic error exponent equals the regularized relative entropy of $\rho$ with respect to the given family of quantum states (for the exact description see \cite{brandao-plenio} and section \ref{subsection-1}). This seminal result was of fundamental importance as it established that any reasonable resource theory (which satisfies Axioms 1–5 listed in section \ref{subsection-1}) can be reversibly manipulated.

However, a gap had been found in the claimed proof of the generalized quantum Stein’s lemma (\cite{gen-stein, gap-2}), which put the main results of Brand\~{a}o and Plenio’s work into question. This also had put into question many results in the literature that were directly or indirectly dependent on this fact, in particular the reversibility of quantum entanglement and reversibility of general quantum resources under asymptotically resource non-generating operations (\cite{brandao-plenio, BP-2, BG}). After finishing this article, we learned that Yamasaki and Kuroiwa \cite{YK} announced a proof for the generalized Stein's lemma. However, there seems to be some issues regarding their proof as explained in the latest version of their article. In a recent breakthrough, the generalized quantum Stein's lemma has been independently proved by Hayashi and Yamasaki (\cite{Hayashi-yamasaki}) and Lami (\cite{Lami-gqsl}).

Before the recent proof of the Stein's lemma appeared, one of the few special cases where the generalized Stein's lemma was known to hold true is the case of quantum coherence, where the second hypothesis consists of incoherent states that are diagonal in a fixed orthonormal basis $\{\ket{i} \}$ of $\mathbb{C}^m$,
\[\mathcal{I}_n:=\operatorname{conv}\Big\{\bigotimes_{k=1}^n \ket{j_k}\bra{j_k}: j_k\in\{1,\ldots,m\}\Big\}.\] In this scenario, by the series of works of Winter-Yang \cite{Winter-Yang}, Zhao et al. (\cite{dilution, distillation}), Chitambar (\cite{Chitambar}), Regula et al. (\cite{regula}), the asymptotic error exponent of testing the i.i.d. copies $\rho^{\otimes n}$ of a fixed state $\rho\in B(\mathbb{C}^m)$ against the sets $\mathcal{I}_{n}$ is proven to be the relative entropy of coherence $C(\rho)=\inf_{\sigma\in \mathcal{I}}D(\rho||\sigma)$. Also, an earlier work of Hayashi-Tomamichel (\cite{HT}) proved the generalized Stein's lemma for a particular subalgebra $\N=1_A\otimes B(\Ha_B)\subset B(\Ha_A\otimes \Ha_B)$.

\subsection{Our Contributions:}
In this work, we establish the generalized Stein's lemma for the state spaces of an arbitrary von Neumann algebra acting on a finite-dimensional Hilbert space. In fact, we prove a stronger result that works for any fixed error parameter $\eps\in (0,1)$, instead of letting $\eps \to 0$.


\begin{theorem}{(Generalized Stein's lemma for subalgebras)}\label{thm:1} Let $\N\subseteq B(\Ha)$ be a von Neumann subalgebra. Then for any $\rho\in B(\Ha)$ and $\eps \in (0,1)$, we have 
\[\lim_{n\to \infty}\frac{1}{n} D_{H}^{\eps}(\rho^{\otimes n}\| \N^{\otimes n})= D(\rho\|\N).\]
\end{theorem}
We refer to the preliminary section \ref{subsection-1} for the precise definition of hypothesis testing relative entropy $D_H^\eps$. This result extends the generalized Stein's lemma beyond the setting of coherence. Note that the above statement implies the strong converse in the sense that it applies to any fixed error parameter $\eps\in (0,1)$, instead of vanishing error $\eps \to 0$. 

\begin{remark}
As mentioned earlier, after finishing this article, we learned that the generalized quantum Stein's lemma has been independently proven by Hayashi and Yamasaki (\cite{Hayashi-yamasaki}) and Lami (\cite{Lami-gqsl}). Despite these developments, we believe that our work provides a different, algebraic approach for the case of subalgebras which can be of independent interest. The novelty of our approach also lies in the fact that we can actually derive a second-order asymptotic analysis for the smoothed max-relative entropy as well as the hypothesis testing relative entropy for the subalgebra setting (see section \ref{subsec:second order}). It should be noted that the papers (\cite{Hayashi-yamasaki, Lami-gqsl}) that resolve Stein's lemma conjecture for general convex sets do not have this feature.  
\end{remark}

Our key ingredient is a dilation theorem (Theorem \ref{theorem: duality}) connecting smoothed subalgebra entropies with smoothed conditional entropies. In addition to establishing the generalized Stein's lemma in the setting of subalgebras, we also explore the resource theory in this scenario. In the context of quantum coherence, free or incoherent states are those diagonal in a priori fixed computational basis. We introduce a notion of generalized coherence with respect to any given subalgebra of a system, where the incoherent states are defined as the states from the given subalgebra of the system. Two fundamental problems in resource theory are dilution and distillation of resourceful states. In this work, we show that the cost of diluting a state $\rho$ from the maximally incoherent state under suitable incoherent operations is given by the subalgebra relative entropy $D(\rho||\N)$. More precisely, we prove the following

\begin{theorem}[Informal]\label{thm 3}
Let $\N\subseteq B(\Ha)$ be a subalgebra. The asymptotic $\N$-coherent cost for a state $\rho$ under maximally incoherent operations is given by $D(\rho\|\N)$.
\end{theorem}

We refer to section \ref{resource theory} for the precise meaning of these quantities.


The rest of the paper is organized as follows. Section 2 reviews the definitions of various entropic quantities, the statement of generalized Stein's lemma, subalgebra entropies as well as the quantum asymptotic equipartition property.  In Section 3, we establish the dilation Theorem \ref{theorem: duality}, and based on that, prove Theorem \ref{thm:1} (in fact, we prove a stronger statement, see Theorem \ref{max-AEP}). Section 4 is devoted to the resource theory of subalgebra coherence where we obtain the asymptotic dilution cost. We conclude the paper with a discussion on several further questions and directions.\\


\noindent {\bf Acknowledgement.}

We thank Bjarne Bergh for helpful discussion regarding the second-order asymptotics of smooth min-conditional entropy. We would like to thank the anonymous referees for carefully reading the manuscript and pointing
out several typos and also suggesting numerous valuable modifications that have
improved the quality and exposition of this work.

LG was partially supported by NSF grant DMS-2154903. This project has received funding from the European Union’s Horizon research and innovation programme under the Marie Sk\l{}odowska-Curie Grant Agreement No. HORIZON-MSCA-2022-PF-01 (Project number: 101108117). This work was partially supported by the Wallenberg Centre for Quantum Technology (WACQT) funded by Knut and Alice Wallenberg Foundation (KAW).
\section{Definitions, notations and Preliminaries }\label{subsection-1}
\subsection{Relative entropies}
Throughout the paper, we consider $\cH,\cK$ to be finite-dimensional Hilbert spaces and denote by $B(\cH)$ the set of linear operators on $\cH$. The standard matrix trace on $B(\cH)$ is denoted as $\tr$. We say $\rho\in B(\Ha)$ is a state (resp. substate) if $\rho$ is positive semidefinite and $\tr(\rho)=1$ (resp. $\tr(\rho)\le 1$). The set of all states (resp. substates) is denoted as $S(\cH)$ (resp. $S_\le (\cH)$).

Let $\rho$ be a substate and $\sigma$ be a positive semidefinite operator in $B(\cH)$ with support projections satisfying $\text{supp}(\rho)\le \text{supp}(\si)$. We recall the definitions of the following relative entropies
\begin{enumerate}
\item[i)] {\bf Relative entropy}: \[D(\rho||\si):=\tr(\rho\log \rho-\rho\log \si)\]

\item[ii)] {\bf Sandwiched R\'enyi relative entropy}: for $\alpha\in [1/2,1)\cup (1,\infty]$, 
\[{D}_\al(\rho||\si):=\frac{1}{\al-1}\log \tr[({\si^{\frac{1-\al}{2\al}}\rho \si^{\frac{1-\al}{2\al}}})^{\al}]. \]

 \item[iii)] {\bf Max relative entropy}:
 \begin{align*}
D_{\max}(\rho||\si)&:=D_{\infty}(\rho||\si)\\
&=\log \inf \{ \pl \la>0 \pl  | \pl \rho \le \la \si \pl \}.
\end{align*}
\item[iv)] {\bf Min relative entropy}: \begin{align*}
D_{\min}(\rho||\si)&:=D_{1/2}(\rho||\si)\\
&=-2\log \tr|\rho^{1/2}\si^{1/2}|\\
&=-\log F(\rho,\si)
\end{align*}
\end{enumerate}
$[v)$ {\bf Hypothesis-testing relative entropy}:
\noindent\begin{align*}D_{H}^\eps(\rho\|\sigma):=-\log\inf\{ \tr(Q\sigma) \pl  | \pl 0\le Q\le 1\pl, \\ \tr(Q\rho)\ge 1-\eps \}. 
\end{align*}

If $\text{supp}(\rho)\nleqslant \text{supp}(\si)$, $D$ and $D_{\max}$ are defined as $+\infty$.
All the logarithms above and in the following are in base $2$. Recall that the purified distance between two substates $\rho $ and $\sigma$ is
\begin{align*}&d(\rho,\sigma)= \sqrt{1-F_*(\rho,\sigma)^2}\pl,
\pl F_*(\rho,\sigma)=\tr(|\sqrt{\rho}\sqrt{\sigma}|)+\sqrt{(1-\tr(\rho))(1-\tr(\sigma))}\pl, \label{eq:pd}\end{align*}
where $F_*(\rho,\sigma)$ is the generalized fidelity which coincides with the fidelity when $\rho$ or $\sigma$ is a state. See Section 3.3 in \cite{tomamichel-book} for more information about purified distance and generalized fidelity.   

For two substates $\rho$ and $\rho'$, we write $\rho'\sim^{\eps} \rho$ if $d(\rho',\rho)\le \eps$. For $0\le \eps\le 1$, we consider the smoothed version of max/min relative entropy,
\begin{enumerate}
\item[vi)] {\bf Smooth max-relative entropy}: \[D_{\max}^\eps(\rho\|\sigma)=\inf_{\rho'\sim^{\eps}\rho}D_{\max}(\rho'\|\sigma),\]
\item[vii)] {\bf Smooth min-relative entropy}: \[D_{\min}^\eps(\rho\|\sigma)=\sup_{\rho'\sim^{\eps}\rho}D_{\min}(\rho'\|\sigma),\]
\end{enumerate}
where the infimum/supremum is over all substate $\rho'$ that is $\eps$-close to $\rho$ in purified distance.

Imagine a source generating several i.i.d. copies of one of two quantum states $\rho, \sigma$, and our task is to decide which one of them is being produced. In order to learn the identity of the state, the observer measures a two-outcome POVM $\{Q_n, 1-Q_n\}$, given $n$ realizations of the unknown state. The state $\rho$ is seen as the null hypothesis, while the state $\sigma$ as the alternative hypothesis. There are two types of errors that one can make:
\begin{itemize}
 \item Type $\mathrm{I}:$ Observer thinks it is $\sigma$, when it is actually $\rho$. This happens with probability $\alpha_n(Q_n)=\tr(\rho^{\otimes n} (1-Q_n))$.

 \vspace{10pt}

 \item Type $\mathrm{II}:$ Observer thinks it is $\rho$, when it is actually $\sigma$. This happens with probability $\beta_n(Q_n)=\tr(\sigma^{\otimes n} (Q_n))$.
\end{itemize}
In the setting of \textbf{asymmetric hypothesis testing}, the probability of type $\mathrm{II}$ error should be minimized, while only requiring that the probability of type $\mathrm{I}$ error is bounded by a small parameter $\eps$. That is,
\[\beta_n(\eps)=\min_{0\leq Q_n\leq 1} \{\beta_n(Q_n): \alpha_n(Q_n)\leq \eps\}.\]

Note that this quantity is closely related to the {\bf hypothesis testing relative entropy} defined in the beginning of the section:
\[-\log \beta_n(\eps)=D_{H}^\eps(\rho^{\otimes n}\|\sigma^{\otimes n}).\]

The celebrated \textbf{Stein's lemma} (\cite{HP, ON}) asserts that for any $\eps\in (0,1)$,
\begin{align}\lim_{n\to\infty} \frac{1}{n}D_{H}^\eps(\rho^{\otimes n}\|\sigma^{\otimes n})=D(\rho\|\sigma). \label{gstein}\end{align}
The limits sharing the same spirit for $D_{\max}^\eps$ and $D_{\min}^\eps$ are called
\textbf{asymptotic equipartition property (AEP)} (see \cite{tomamichel-book, tomamichel, toma-thesis, nuradha2024fidelity})
\begin{align*}&\lim_{n\to\infty} \frac{1}{n}D_{\max}^\eps(\rho^{\otimes n}\|\sigma^{\otimes n})=\lim_{n\to\infty} \frac{1}{n}D_{\min}^\eps(\rho^{\otimes n}\|\sigma^{\otimes n})
=D(\rho\|\sigma).\pl \label{saep}\end{align*}


\subsection{Generalized Stein's lemma}

In the framework of composite hypothesis testing, the null hypothesis or the alternative hypothesis is given by a more general set of quantum states. In fact, when distinguishing $\rho$ from a whole set of quantum states, one considers this generalized Stein's Lemma scenario, where the alternative hypothesis is a convex set of quantum states. As we are going to distinguish many copies $\rho^{\otimes n}$ of $\rho$,  from copies of the given set, we need a hierarchy of the sets of alternative hypotheses that grow with $n\geq 1$. 
The family of states satisfies some properties that are motivated from a general resource theory point of view, and we outline them below.
Let ${\cF}_n \subseteq {\cS}({\cH}^{\otimes n}), n \in \mathbb{N}$, satisfying the following properties:
\begin{enumerate}
        \item \label{cond1} Each ${\cF}_n$ is convex and closed.
        \item \label{cond2} Each ${\cF}_n$ contains $\sigma^{\otimes n}$,
        for a full rank state $\sigma \in {\cS}({\cH})$.
        \item \label{cond3} If $\rho \in {\cF_{n +1}}$, then
        $\tr_{k}(\rho) \in {\cF}_{n}$, for every $k \in \{1, ..., n + 1 \}$.
        \item \label{cond4} If $\rho \in {\cF}_{n}$ and
        $\nu \in {\cF}_m$, then $\rho \otimes \nu \in
        {\cF}_{n+m}$.
        \item \label{cond5} If $\rho \in {\cF}_n$, then
        $P_{\pi}\rho P_{\pi}^* \in {\cF}_n$; where $\pi \in S_n$ is a permutation of $n$ elements and $P_\pi$ is the unitary implementing the permutation $\pi$ on $\Ha^{\otimes n}$.

\end{enumerate}

The above axioms are usually inspired by quantum resource theories. Indeed, the general framework of quantum resource theories (\cite{CG}) is designed to analyze various quantum resources under a unified format, the set of entangled states being one of the resources. Here, one usually identifies a family of systems of interest, modeled by finite-dimensional Hilbert space $\cH$, and over each $\cH$ a set ${\cF}_n \subseteq {\cS}({\cH}^{\otimes n}), n \in \mathbb{N}$ of free states, i.e. quantum states that are easily available, for example, the separable states.

Apart from the separable states, another example that satisfies the above axioms is the set of incoherent states \[\F_n=\displaystyle\rm{conv} \{\bigotimes_{k=1}^n |j_k\rangle\langle j_k|  \},\]
for a fixed orthonormal basis $\{|j\rangle\}$ of $\Ha$.

Now given ${\cF}_n \subseteq {\cS}({\cH}^{\otimes n}), n \in \mathbb{N}$, the \textbf{relative entropy of resource} of any $\rho\in \S(\Ha)$ is defined as

\[D_\F(\rho)=\displaystyle \inf_{\sigma\in \F_1} D(\rho\| \sigma),\]
and the regularized \textbf{relative entropy of resource} is defined as

\[D_\F^\infty(\rho)=\displaystyle \lim_{n\to\infty} \frac{1}{n}D_{\F_n}(\rho^{\otimes n}).\]

In this context, the generalized Stein's lemma is stated as follows.

\begin{theorem} [Generalized quantum Stein's lemma (see \cite{Hayashi-yamasaki, Lami-gqsl})]
For any family of sets $\{ {\cF}_n \}_{n \in \mathbb{N}}$ satisfying the Axioms 1,2 and 4 in subsection \ref{subsection-1}, it holds that for any $ \varepsilon\in (0,1)$
\begin{equation}\label{eq-gen-stein}
 \liminf_{n\to\infty} \frac{1}{n} \min_{\sigma_n\in \F_n} D_{H}^\eps(\rho^{\otimes n}\| \sigma_n)=D_\F^\infty(\rho).
\end{equation}
\end{theorem}

As stated earlier, before the recent breakthrough, where the proof of the generalized Stein's lemma was put forward, there were few instances where the validity of this result was known to be true. These cases include quantum coherence, pseudo-entanglement, etc (see \cite{gen-stein}). In the context of quantum coherence with $\N$ being the diagonal subalgebra, this was established in \cite{BBH} (see also \cite{max-coherence} \cite{distillation}). Another special case, where $\N=\mathbb{C}{\bf 1}_A\otimes B(\cH_{B})\subseteq B(\cH_A)\otimes B(\cH_B)$, for which the generalized Stein's lemma is known to be true was shown by \cite{HT}. Thus, our result captures all these known cases and provides a larger class of examples where the generalized Stein's lemma can be established. Another special case, whereles where the generalized Stein's lemma can be established. We again emphasize that our work provides a finer analysis of the composite hypothesis testing in the subalgebra case, as it allows one to obtain a second-order expansion of the hypothesis testing relative entropy.


Let $\M\subseteq B(\Ha)$ be the commutative subalgebra spanned by the fixed basis $\{|j\rangle\}$. The incoherent states
$\rm{conv}\{\bigotimes_{k=1}^n |j_k\rangle\langle j_k|\}$ can be viewed as the state spaces of the commutative subalgebra $\M^{\ten n}$. Indeed, this is a special case of a more general fact: given any von Neumann subalgebra  $\N\subseteq B(\Ha)$, the state spaces
$\F_n=\S(\N^{\otimes n})\subseteq \S(\Ha^{\otimes n})$ satisfy the above axioms, where by $\S(\N^{\otimes n})$ we denote $S(\Ha^{\otimes n})\cap \N^{\otimes n}$. Here we note that $\S(\N^{\otimes n})$ may contain separable states as well as entangled states. Indeed, if the subalgebra $\N\subseteq B(\Ha)$ is non-commutative, then it is not hard to see that the state space $\S(\N^{\otimes n})$ contains both separable and entangled states of the system $\S(\Ha^{\otimes n})$. This is interesting, as in the quantum coherence, one always obtains separable states as a resource.


\subsection{Subalgebra entropies}\label{subsection: subalgebra entropy}
Let $\cH$ be a finite-dimensional Hilbert space and $\N\subset B(\cH)$ be a von Neumann subalgebra (or simply, a unital $*$-subalgebra).
The state space of $\N$ is defined as
\[ \S(\N)= \N\cap \S(\Ha)\pl.\]
Given a state $\rho\in B(\Ha)$, we define the subalgebra entropy
\[ \mathbb{D}(\rho\|\N):=\inf_{\sigma\in \S(\N)}  \mathbb{D}(\rho\|\sigma).\]
for $\mathbb{D}=D, D_\al,  D_H^\eps, D_{\max}, D_{\min}$  correspondingly. Some of these quantities were studied in \cite{GJL19}. Also, the smooth max-relative entropy with respect to $\N$ is defined as 
\[\label{minmaxN} D_{\max}^\eps(\rho\|\N)=\inf_{\rho'\sim^{\eps}\rho}D_{\max}(\rho'\|\N),\]
and the smooth min-relative entropy with respect to $\N$ is defined as 
\[D_{\min}^\eps(\rho\|\N)=\sup_{\rho'\sim^{\eps}\rho}D_{\min}(\rho'\|\N).\]
The above definitions appeared for the special case of the diagonal subalgebra in \cite{max-coherence} . In particular, it is known (see \cite{GJL19}, Proposition 2.1) that for the relative entropy $\mathbb{D}=D$,
\[D(\rho\|\N)= D(\rho\| \E_\N(\rho))\]
where $\E_\N:B(\cH)\to B(\cH)$ is the unique map onto $\N$ such that for all $x\in\N, y\in B(\cH)$
\begin{align} \tr(xy)=\tr(x\E_\N(y))\pl. \label{cd}\end{align}
This map $\E_\N$ is called the (trace-preserving) conditional expectation onto $\N$.
Moreover, $\E_\N$ is also unital completely positive, hence a quantum channel (CPTP map), sending identity to identity.

The Stein's lemma (Theorem \ref{thm:1}) for hypothesis testing relative entropy $D_H^\eps$ follows from the following strong asymptotic equipartition property (AEP) for $D_{\max}^\eps$ (see the proof of Theorem \ref{max-AEP}).
\begin{theorem}[AEP for subalgebras] \label{main-thm-2} For any $\rho\in B(\Ha)$, subalgebra $\N\subseteq B(\Ha)$ and $\eps \in (0,1)$, we have
\begin{align*}
\lim_{n\to \infty}\frac{1}{n}D_{\max}^{\eps}(\rho^{\otimes n} \| \N^{\otimes n})=\lim_{n\to \infty}\frac{1}{n}D_{\min}^{\eps}(\rho^{\otimes n} \| \N^{\otimes n})= D(\rho\|\N).
\end{align*}
\end{theorem}
Note that here the right hand side does not need regularization because
\begin{align*} D(\rho^{\otimes n}\| \N^{\otimes n})=D(\rho^{\otimes n}\| \E_{\N^{\otimes n}}(\rho^{\otimes n}))=D(\rho^{\otimes n}\| \E_{\N}^{\otimes n}(\rho^{\otimes n}))
=nD(\rho\|\N)\pl.
\end{align*}

We also emphasize that the above theorem works for any fixed $\eps \in (0,1)$, and it is crucial in deriving the Stein's lemma.

We illustrate the above theorem with some special cases.

\begin{exam}{\label{exam:1} \rm  Let $\N=\mathbb{C} {\bf 1}$ be the trivial subalgebra. The conditional expectation onto the scalars $\mathbb{C} {\bf 1}\subset B(\Ha)$ is the completely depolarizing map
\[\E_{\mathbb{C}}:B(\Ha)\to \mathbb{C}{\bf 1}\pl, \E_{\mathbb{C}}(\rho)=\tr(\rho)\frac{\bf 1}{d}\pl, \]
where $d=\dim \Ha$. In this setting, the subalgebra entropy is essentially the  von Neumann entropy as \[D(\rho\|\mathbb{C} {\bf 1})=\log d- H(\rho)\pl, \]
where $H(\rho)=-tr(\rho\log \rho )$ is the von Neumann entropy. 

Similarly,
\begin{align*}D_{\max}(\rho \| \mathbb{C} {\bf 1})=\log d- H_{\min}(\rho)\pl, \\
D_{\min}(\rho \| \mathbb{C} {\bf 1})=\log d- H_{1/2}(\rho)\pl,
\end{align*}
where \begin{align*}&H_{\min}(\rho)=H_{\infty}(\rho)=-\log \norm{\rho}{\infty},\\
&H_{1/2}(\rho)=2\log \tr(\sqrt{\rho}).
\end{align*}
Theorem \ref{main-thm-2} gives the AEP for von Neumann entropy,
\[\lim_{n\to \infty}\frac{1}{n}H_{\max}^{\eps}(\rho^{\otimes n} )=\lim_{n\to \infty}\frac{1}{n}H_{\min}^{\eps}(\rho^{\otimes n} )= H(\rho).\]}
\end{exam}
\begin{exam}{\rm \label{exam:2}Let $\Ha_{A}$ and $\Ha_{B}$ be two Hilbert spaces.
The following AEP for conditional entropy was proved in \cite{tomamichel, tomamichel-book} that for a bipartite state $\rho_{AB}\in \S(\Ha_{A}\ten \Ha_{B})$,
\begin{align}\label{eq:qaep} &\lim_{n\to \infty}\frac{1}{n}H_{\max}^{\eps}(A^n|B^n)_{\rho^{\otimes n}}=\lim_{n\to \infty}\frac{1}{n}H_{\min}^{\eps}(A^n|B^n)_{\rho^{\otimes n}}= H(A|B)_\rho\pl.\end{align}
where the conditional entropies are defined as
\begin{align*}
&H(A|B)_\rho=H(\rho_{AB})- H(\rho_B)\pl,\\
&H_{\max}(A|B)_{\rho}=-\inf_{\sigma \in \S(\Ha_B)}D_{\min}(\rho_{AB}||  {\bf 1}_A\ten \sigma_B)\pl,\\
&\pl H_{\max}^\eps(A|B)_{\rho}=\min_{\rho\sim^{\eps} \rho'}H_{\max}(A|B)_{\rho'}.\\
&H_{\min}(A|B)_{\rho}=-\inf_{\sigma \in \S(\Ha_B)}D_{\max}(\rho_{AB}|| {\bf 1}_A\ten \sigma_B)\pl,\\
&\pl H_{\min}^\eps(A|B)_{\rho}=\max_{\rho\sim^{\eps} \rho'}H_{\min}(A|B)_{\rho'}
\end{align*}
Up to the same normalization of the previous example, this is special case of Theorem \ref{main-thm-2} by choosing $\N=\mathbb{C}{\bf 1}_A\ten B(\Ha_{B})$ being a tensor component of $B(\Ha_{A}\ten \Ha_{B})$.
In this case, the conditional expectation is defined as
$X_A\otimes X_B\mapsto \mathrm{tr}(X_A)({\bf 1}_A/d_A)\otimes X_B$, where $d_A$ is the dimension of the $A$-system.

} 
\end{exam}


\begin{exam}{\rm \label{exam:3} For quantum coherence, let $\{\ket{j}\}_{j=1}^d$ be a fixed orthonormal basis of $\Ha$. The commutative subalgebra $\I=\text{span}\{\ket{j}\bra{j}\}_{j=1}^d$ models the incoherent subalgebra and the density operators in $\I$ are states with no coherence, called incoherent states. The recent progress establishes that the relative entropy of coherence is the fundamental quantity quantifying the cost of asymptotic conversion between coherent states (\cite{dilution, distillation, Winter-Yang, HFW}). 
This result can be understood as a special case of our theorem \ref{main-thm-2}, where $\N=\I$ is the maximal abelian subalgebra of $B(\ha)$.}
Note that here $D(\rho||\N)$ is the relative entropy of coherence. Both limits hold in Theorem \ref{main-thm-2} for a fixed error parameter $\eps$ without letting $\eps \to 0$, which has implications in the \textbf{strong converse} of the resource theory of coherence. Here, a strong converse usually means that the limit equation in the theorem \ref{main-thm-2} holds for any $\eps\in (0, 1)$. Such strong converse was also obtained in the recent work of Hayashi, Fang, and Wang (\cite{HFW}).
\end{exam}

\begin{exam}{\rm \label{exam:4} Let $G$ be a compact group and $U: G \to B(\Ha)$ be a unitary representation. Consider $\N_G:=\{\pl x \in B(\Ha)\pl | \pl U_gxU_g^*=x \pl,  \forall \pl g\in G \pl. \}$ be the $G$-invariant subalgebra. The conditional expectation onto $\N_G$ is given by
\[\E_G(x)=\int_G U_gxU_g^*d\mu(g)\pl, \]
where $\mu$ is the Haar measure. For any state $\rho$, $\E_G(\rho)$ is the projection of $\rho$ to the $G$-symmetric states. The relative entropy $D(\rho||\N_G)=D(\rho||\E_G(\rho))$ is called the relative entropy of $G$-asymmetry, and $D_{\min}(\rho||\N_G)$ and $D_{\max}(\rho||\N_G)$ are called $G$-asymmetry measures in \cite{Marvian-thesis, MS}. Our above theorem implies that for any compact group symmetry,
\begin{align*} &\lim_{n\to \infty}\frac{1}{n}D_{\max}^{\eps}(\rho^{\otimes n} \| \N_G^{\otimes n})=\lim_{n\to \infty}\frac{1}{n}D_{\min}^{\eps}(\rho^{\otimes n} \| \N_G^{\otimes n})= D(\rho\|\N_G)\pl.
\end{align*}
}
\end{exam}

\section{Proof of Stein's lemma for subalgebra entropies}\label{subalge vs conditional entropy}

\subsection{Conditional expectations:}
We need some preliminary results to get to the proof of the main theorems in this section. 
Specifically, we need some results on conditional expectations, which can be of independent interest. 
We start with a lemma about the multiplicative domain of a unital completely positive (UCP) map.  Let $\Phi:\A\to\B$ be a UCP map between two C$^*$-algebras $\A$ and $\B$. The multiplicative domain of $\Phi$ is defined as follows:
\begin{align*}\M_\Phi= \{a\in \A: \Phi(ax)=\Phi(a)\Phi(x),
\ \Phi(xa)=\Phi(x)\Phi(a), \forall x\in \A\}.\end{align*}
\begin{lemma}\label{lemma-comm}
Let $\Phi: \A\rightarrow B(\mathcal{H})$ be a UCP map with Stinespring dilation $\Phi(x)=V^*\pi(x)V$, where $\pi: \A\rightarrow B(\K)$ is a homomorphism  and $V: \mathcal{H}\rightarrow \K$ is a Hilbert space isometry. Then we have
\[\M_\Phi=\{a\in \A: \pi(a) VV^*=VV^*\pi(a).\}\]
\end{lemma}
\begin{proof}
Note that as $V$ is an isometry, $V^*V=1$ and $VV^*=e$ is a projection. It is well known (due to Choi \cite{choi1}) that \begin{align*}\M_\Phi=\{a\in \A: \Phi(aa^*)=\Phi(a)\Phi(a)^*, \ \Phi(a^*a)=\Phi(a^*)\Phi(a)\}.\end{align*}
Using this characterization, one way is immediate:
if $a\in \A$ such that $\pi(a)e=e\pi(a)$, then, it holds that
\begin{align*}\Phi(a)\Phi(a^*)=V^*\pi(a)VV^*\pi(a^*)V
=V^*VV^*\pi(a)\pi(a^*)V
=\Phi(aa^*).
\end{align*}
The other side $\Phi(a^*a)=\Phi(a^*)\Phi(a)$ follows similarly, hence $a\in \M_{\Phi}$.

Conversely, suppose $a\in \M_\Phi$. From $\Phi(aa^*)=\Phi(a)\Phi(a^*)$ we get
\[ V^*\pi(a)\pi(a^*)V=V^*\pi(a)VV^*\pi(a^*)V.\]
Then \begin{align*}e\pi(a)\pi(a^*)e &=VV^*\pi(a)\pi(a^*)VV^*\\
&=V ( V^*\pi(a)VV^*\pi(a^*)V) V^* \\
&=e\pi(a)e\pi(a^*)e.
\end{align*}
Hence we get
\begin{align*} 0&=e\pi(a)\pi(a^*)e-e\pi(a)e\pi(a^*)e\\
&=e\pi(a)(1-e)\pi(a^*)e\\
&=e\pi(a)(1-e) [e\pi(a)(1-e)]^*
\end{align*}
This yields, \begin{equation}\label{eq-1}
e\pi(a)(1-e)=0\implies e\pi(a)=e\pi(a)e.
\end{equation}
Since the multiplicative domain is closed under adjoint, we also have 
\begin{equation}\label{eq-2}e\pi(a^*)=e\pi(a^*)e.\end{equation}
Putting \eqref{eq-1} and \eqref{eq-2} together,
\begin{align*}e\pi(a)&=e\pi(a)e=(e\pi(a^*)e)^*=(e\pi(a^*))^*=\pi(a)e,
\end{align*}
which completes the proof.
\end{proof}
\begin{remark}{\rm 
Note that if $\Phi(A)=\sum_i K_i^* XK_i$ is a ucp map on a finite-dimensional Hilbert space $\Ha$, then the Stinespring isometry can be defined as $V=\sum_i  e_i\otimes K_i$ and the homomorphism $\pi(A)={\bf 1}\otimes A$. Then the multiplicative domain is the commutant of the set $S:=\text{span}\{K_iK_j^*\}_{i,j}$ (see Proposition 2 in \cite{carbone-jencova}). Note that this set is not unital but closed under the adjoint and a closely related set $T:=\text{span}\{K_i^*K_j\}_{i,j}$ is an important set, called the confusability graph of the channel $\Phi^\dagger(\cdot)=\sum_i K_i \cdot K_i^*$ in the theory of zero-error communication (\cite{DSW}). Note also that the commutant of $T$ was identified as the largest correctable subalgebra of observables in the Heisenberg picture in \cite{BKK}. For unital channels, the multiplicative domain is the commutant of $T$ (see Corollary 1.4 in \cite{miza}).} 
\end{remark}
We apply the above lemma to conditional expectations. Recall (see \cite{takesaki}) that a finite von Neumann algebra $\M$ is a von Neumann algebra equipped with a normal, faithful finite trace, we denote it by $\tr$. Here, finite means that the trace of the identity element is finite, $\tr({\bf 1})<\infty$. By faithfulness, we mean that if $\tr(a)=0$, for any positive element $a\in \M$, then $a=0$. And finally, a trace is normal, if it is weak$^*$ continuous, that is, if for every monotone increasing net of elements $\{x_\alpha\}$ in $\M$ with least upper bound $x$, we have $\tr(x_\alpha)$ converges to $\tr (x)$. It is known that in a finite von Neumann algebra, every von Neumann subalgebra admits a unique trace-preserving conditional expectation (see \cite{Umegaki}). 
\begin{prop}\label{prop-cond. exp}
Let $\N\subseteq \M$ be a subalgebra of a finite von Neumann algebra $\M\subset B(\mathcal{H})$ and
$\E: \M \rightarrow \N$ be a trace-preserving conditional expectation with a Stinespring representation
$\E(x)=V^* \pi(x) V$,
where $\pi: \M\rightarrow B(\K)$ is a homomorphism and $V: \mathcal{H} \rightarrow \K $ is an isometry. Then it holds that \[V\E(x)V^*\leq \pi (\E(x)), \forall x\in \M_+, \]
where $\M_+$ denotes the set of all positive elements of $\M$.
\end{prop}
\begin{proof}It is easy to see that the multiplicative domain $\M_\E$ of $\E$ contains $\N$. Indeed, if $a\in \N$, then so is $a^*$, and hence $aa^*\in \N$. Thus, $\E(aa^*)=aa^*=\E(a)\E(a^*)$, which, by Choi's multiplicative domain characterization, shows that $a\in \M_\E$.
This shows $\N\subset \M_\E$. 

Now, consider a Stinespring dilation $\E(\cdot)=V^* \pi(\cdot) V$. From Lemma \ref{lemma-comm}, it holds that $$\pi(\E(x))VV^*=VV^*\pi(\E(x)), \forall x\in \M.$$
Thus for any $x\in \M_+$,  we have
\begin{align*}\pi(\E(x))\geq VV^*\pi(\E(x)) VV^*=V\E(\E(x))V^*=V\E(x)V^*.\end{align*}
where we have used that $\E\circ\E=\E$.
\end{proof}

\begin{remark}{\rm \label{re:order}
For any ucp map $\Phi:B(\mathcal{H})\rightarrow B(\mathcal{H})$, the Stinespring dilation can be
\[\Phi(a)=V^* (a\otimes {\bf 1}_E)V,\]
where $V: \mathcal{H}\rightarrow \mathcal{H}\otimes \Ha_E$ is an isometry and the homomorphism $\pi$ is simply $a\rightarrow a\otimes 1_E$. Using the above proposition, we have 
\[V\E(x)V^*\leq \E(x)\otimes {\bf 1}_E, \forall x\in \M_+. \]
}\end{remark}

Before we proceed further, we note here that any trace-preserving conditional expectation is self-adjoint, that is, $\E=\E^\dagger$. Here, for $\E: B(\Ha)\rightarrow \N$, its Hilbert-Schmidt adjoint $\E^\dagger$ is defined as an operator on $B(\Ha)$ by simply assigning the value $0$ to the orthogonal complement of $\N$. And hence we have 
\begin{equation}\label{eq-selfdual}
\E(x)=V^*(x\otimes {\bf 1}_E)V=\E^\dagger(x)=\tr_E(VxV^*)
\end{equation}

To see this, for any $x, y$ we have
\begin{align*}\langle \E(x), y\rangle&=\tr (\E(x)y^*)=\tr(\E(\E(x)y^*))=\tr(\E(x)\E(y^*)), 
\end{align*}
where we used the fact that the range of $\E$ is inside the multiplicative domain and $\E$ is a trace-preserving idempotent. On the other hand, using similar properties we have
\begin{align*}\langle \E^\dagger(x), y\rangle=\tr (\E^\dagger(x)y^*)=\tr (x\E(y^*))=\tr(\E(x\E(y^*)))=\tr(\E(x)\E(y^*)).
\end{align*}
 Hence $\E=\E^\dagger$.
\subsection{A dilation result and generalized Stein's lemma}

The main result of this section connects the relative entropies for subalgebras to conditional entropies of an appropriate state. Since conditional entropies are widely used in quantum cryptography and quantum information processing tasks, our results provide an operational meaning to the relative entropies of subalgebras. 
Recall that for a bipartite state $\rho_{AB}\in \S(\Ha_A\otimes \Ha_B)$, and for $\alpha\in [1/2,1)\cup (1, \infty)$, the sandwiched R\'{e}nyi conditional entropy is defined as
\[H_\al (A|B)_{\rho_{AB}}=\sup_{\sigma_B\in \S(\Ha_B)} -D_\al(\rho_{AB}|| {\bf 1}_A\otimes \sigma_B), \]
and it follows that (see \cite{tomamichel-book}) \begin{align*}\lim_{\al\to 1} H_\al(A|B)_{\rho_{AB}}&=H(A|B)_\rho\end{align*}
where the standard relative entropy $(H(A|B)_\rho)$ is defined as
\begin{align*}
H(A|B)_\rho&=H(AB)_\rho-H(B)_\rho\\
&=-\tr (\rho_{AB}\log \rho_{AB}) + \tr(\rho_B\log \rho_B).\end{align*}
We are now ready to prove our dilation theorem. We refer to section \ref{subsection-1} for the definitions of  smoothed relative and conditional entropies.

\begin{theorem}\label{theorem: duality}
Let $\Ha_A$ be a finite-dimensional Hilbert space and $\N\subset B(\Ha_A)$ be a subalgebra. Suppose $V: \Ha_A\rightarrow \Ha_A\otimes \Ha_E$ is a Stinespring dilation isometry for the trace-preserving conditional expectation $\E: B(\Ha_A)\rightarrow \N$ such that
$\E(x)=V^*(x\otimes {\bf 1}_E)V$
\begin{itemize}
\item[i)] for $\eps\in [0,1)$,
\begin{align*}&D_{\max}^\eps(\rho\| \N)= - H_{\min}^\eps (E|A)_{V\rho V^* }.\\
&D_{\min}^\eps(\rho\| \N)= - H_{\max}^\eps (E|A)_{V\rho V^* }.
\end{align*}
\item[ii)]for $\al\in [1/2, \infty]$,
\[D_\alpha(\rho\| \N)= - H_\al (E|A)_{V\rho V^* }.\]
\item[iii)] Given a bipartite state $\rho_{AB}$,  define the hypothesis testing conditional entropy 
\[ H_h^\eps (A|B)_{\rho }:=-\inf_{\sigma_B\in\S(\Ha_B)} D_H^\eps(\rho_{AB}||{\bf 1}_A\ten \sigma_B)\pl.\] Then we have for any state $\rho$,
\[D_H^{\eps}(\rho\| \N)= - H_h^\eps (E|A)_{V\rho V^* }.\]
\end{itemize}
\end{theorem}

\begin{proof}
 We first prove the case for $D_{\max}^\eps$ and $D_{\min}^\eps$, which correspond to $\al=\infty$, and $\al=1/2$ (respectively) in the Sandwiched R\'enyi relative entropy $D_\al$. The argument for the rest, that is, for $D_\al, \al\in (1/2, \infty)$ is given in the end of the proof although it follows similar lines of arguments. 

 Denote by $e=VV^*$ the projection in $B(\Ha_A\ten \Ha_E)$.  For $D^\eps_{\max}$,
\begin{align*}
 D_{\max}^\eps(\rho\| \N)
 =&\inf_{\rho'\sim^{\eps}\rho} \inf_{\sigma\in \S(\Ha_A)} D_{\max}(\rho'\| \E(\sigma))\\
 \overset{(1)}{=}&\inf_{\rho'\sim^{\eps}\rho}\inf_{\sigma\in \S(\Ha_A)} D_{\max}(V\rho' V^*\| V V^*(\sigma\otimes {\bf 1}_E)VV^*)
  \\= &\inf_{\rho'\sim^{\eps}\rho} \inf_{\sigma\in \S(\Ha_A)} D_{\max}(V\rho' V^*\| e(\sigma\otimes {\bf 1}_E)e)
  \\ \overset{(2)}{=} &\inf_{\omega'\sim^{\eps} V\rho V^*} \inf_{\sigma\in \S(\Ha_A)} D_{\max}( e\omega'e \| e(\sigma\otimes {\bf 1}_E)e).
 \end{align*}
 Here the equality (1) follows from the fact that the isometry $V$ preserves $D_{\max}$ (see section $\mathrm{III}$ in \cite{purified-distance}). The equality (2) can be justified as follows:
 given an element $\omega'$ that is $\epsilon$-close to $V\rho V^*$, we have $$V\rho V^*=eV\rho V^*e\sim^{\epsilon}e\omega'e=V(V^*e\omega'e V)V^*,$$ where we used the fact that purified distance is monotone under trace non-increasing CP maps, that is, 
 \[ d(V\rho V^*, e\omega'e) \leq d(V\rho V^*, \omega').\]
  Then defining $\rho':=V^*e\omega'e V\sim^{\epsilon}V^*(V\rho V^*)V=\rho$, we have  $$D_{\max}(V\rho'V^* \| e(\sigma\otimes {\bf 1}_E)e)=D_{\max}(e\omega'e \| e(\sigma\otimes {\bf 1}_E)e).$$ 
 This shows one direction of (2). The other direction follows by taking $\omega'=V\rho'V^*$ for each feasible $\rho'$. 
 
  By the fact that the restriction map $\omega\mapsto e\omega e$ does not increase the purified distance and $D_{\max}$, we have the following inequality
 \begin{align*}
 D_{\max}^\eps(\rho\| \N)
 &{\leq} \inf_{\omega'\sim^{\eps} V\rho V^*} \inf_{\sigma\in \S(\Ha_A)} D_{\max}(\omega' \| \sigma\otimes {\bf 1}_E)\\
 &{\leq} \inf_{\omega'\sim^{\eps} V\rho V^*} \inf_{\sigma\in \S(\N)} D_{\max}(\omega' \| \sigma\otimes {\bf 1}_E)\\
  &\overset{(3)}{\le } \inf_{\rho'\sim^{\eps}\rho}  \inf_{\sigma\in \S(\N)} D_{\max}(V\rho' V^* \| \sigma\otimes {\bf 1}_E)\\
  &\overset{(4)}{\le } \inf_{\rho'\sim^{\eps}\rho}  \inf_{\sigma\in \S(\N)} D_{\max}(V\rho' V^* \| V\sigma V^*)\\
  &=\inf_{\rho'\sim^{\eps}\rho} \inf_{\sigma\in \S(\N)} D_{\max}(\rho' \| \sigma)\\
  &=D_{\max}^{\eps}(\rho\| \N)
  \end{align*}
  The inequality (3) is by defining $\omega'=V\rho'V^*$, for any $\rho'\sim^{\eps} \rho$. The inequality (4) uses the Lemma \ref{lemma-comm} and Remark \ref{re:order} that for $\sigma\in \N_+$
  \[ [e, \sigma\otimes {\bf 1}_E]=0\pl ,\pl  V\sigma V^*\le \sigma\otimes {\bf 1}_E,\]
  and the fact that $V\sigma V^*\le \sigma\otimes {\bf 1}_E$, implies $D_{\max}(\cdot \| \sigma\otimes {\bf 1}_E)\leq D_{\max}(\cdot \| V\sigma V^*)$ (\cite{tomamichel-book}).
Hence we have equality throughout, yielding
\begin{align*} D_{\max}^\eps(\rho\| \N)=\inf_{\omega'\sim^{\eps} V\rho V^*} \inf_{\sigma\in B(\Ha_A
)} D_{\max}(\omega' \| \sigma\otimes {\bf 1}_E)=-H_{\min}^\eps(E|A)_{V \rho V^*}. \end{align*}
For $D^\eps_{\min}$, we have
\begin{align*}
 D_{\min}^\eps(\rho\| \N)=&\sup_{\rho'\sim^{\eps}\rho} \inf_{\sigma\in \S(\Ha_A)} D_{\min}(\rho'\| \E(\sigma))\\
 =&\sup_{\rho'\sim^{\eps}\rho}\inf_{\sigma\in \S(\Ha_A)} D_{\min}(V\rho' V^*\| V V^*(\sigma\otimes {\bf 1}_E)VV^*)\\
  =& \sup_{\rho'\sim^{\eps}\rho} \inf_{\sigma\in \S(\Ha_A)} D_{\min}(V\rho' V^*\| e(\sigma\otimes {\bf 1}_E)e)\\
  \overset{(1)}{= }& \sup_{\rho'\sim^{\eps}\rho} \inf_{\sigma\in \S(\Ha_A)} D_{\min}(V\rho' V^*\| \sigma\otimes {\bf 1}_E)
  \end{align*}
  Here, the equality (1) is because $V\rho' V^*$ is supported on $e$, and the fact that
\begin{align}& D_{\min}(\rho||\sigma)=-2\log F(\rho,\sigma)\pl, \ \text{and} \ \pl F(e\rho e,\tau)= F(e\rho e, e\tau e)=F(\rho, e\tau e),. \label{eq:1}\end{align}
for any positive elements $\rho, \tau$. 

Then by having $\omega':=V\rho' V^*\sim^{\eps} V\rho V^*$,
  \begin{align*}
  &D_{\min}^\eps(\rho\| \N)\\
  \leq &\sup_{\omega'\sim^{\eps} V\rho V^*} \inf_{\sigma\in \S(\Ha_A)} D_{\min}(\omega'\| \sigma\otimes {\bf 1}_E)\\
   \leq &\sup_{\omega'\sim^{\eps} V\rho V^*} \inf_{\sigma\in \S(\N)} D_{\min}(\omega'\| \sigma\otimes {\bf 1}_E)\\
 \overset{(2)}{\leq}& \sup_{\omega'\sim^{\eps} V\rho V^*} \inf_{\sigma\in \S(\N)} D_{\min}( \omega' \| e(\sigma\otimes {\bf 1}_E)e )\\
\overset{(3)}{=}&\sup_{\omega'\sim^{\eps} V\rho V^*} \inf_{\sigma\in \S(\N)} D_{\min}( e\omega'e \| e(\sigma\otimes {\bf 1}_E)e )\\
\overset{(4)}{\leq} &\sup_{e\omega'e\sim^{\eps} V\rho V^*}  \inf_{\sigma\in \S(\N)}  D_{\min}(V V^*\omega'V V^* \| V V^*(\sigma\otimes {\bf 1}_E)V V^*)\\
  = &\sup_{\rho'\sim^{\eps}\rho}  \inf_{\sigma\in \S(\N)} D_{\min}(V\rho' V^* \| V\sigma V^*)\\
  \overset{(5)}{=}&\sup_{\rho'\sim^{\eps}\rho} \inf_{\sigma\in \S(\N)} D_{\min}(\rho' \| \sigma)\\
  =&D_{\min}^{\eps}(\rho\| \N)
  \end{align*}
Here, the inequality (2) uses again the  $[e, \sigma\otimes {\bf 1}_E]=0$ and hence $e(\sigma\otimes {\bf 1}_E) e\le \sigma\otimes {\bf 1}_E$. The equality (3) follows from \eqref{eq:1}. The inequality (4) follows from the fact that the purified distance satisfies the property $d(e\omega'e,V\rho V^*) \le d(\omega', V\rho V^*)$. The equality (5) uses $V^*(\sigma\otimes {\bf 1}_E)V=\E(\sigma)=\sigma$ for $\sigma\in \N$.

Hence we have equality throughout, yielding
\begin{align*} D_{\min}^\eps(\rho\| \N)
= \sup_{\omega'\sim^{\eps} V\rho V^*} \inf_{\sigma\in \S(\Ha_A)} D_{\min}(\omega' \| \sigma\otimes {\bf 1}_E)=-H_{\max}^\eps(E|A)_{V \rho V^*}. 
\end{align*}

For the sandwiched R\'enyi entropy case $D_\al, \al\in [1/2, \infty]$, it follows easily:
 \begin{align*}
 D_\alpha(\rho\| \N)
 =&\inf_{\sigma\in \S(\Ha_A)} D_\alpha(\rho\| \E(\sigma))\\
 =&\inf_{\sigma\in \S(\Ha_A)} D_\alpha(V\rho V^*\| V V^*(\sigma\otimes {\bf 1}_E)VV^*)\\
 = &\inf_{\sigma\in \S(\Ha_A)} D_\alpha(V\rho V^*\| e(\sigma\otimes {\bf 1}_E)e)\\
 \leq &\inf_{\sigma\in \S(\Ha_A)} D_\alpha(V\rho V^* \| (\sigma\otimes {\bf 1}_E))\\
 \leq &\inf_{\sigma\in \S(\N)} D_\alpha(V\rho V^* \| (\sigma\otimes {\bf 1}_E))\\
  \leq & \inf_{\sigma\in \S(\N)} D_\alpha(V\rho V^* \| V\sigma V^*)\\
  =& \inf_{\sigma\in \S(\N)} D_\alpha(\rho \| \sigma)\\
  =&D_\alpha(\rho\| \N),
  \end{align*}
where the equalities hold due to the invariance of $D_\al$ under isometries (\cite{tomamichel-book}). Recall that Sandwiched R\'enyi relative entropy for $\alpha\in (1,\infty)$ is monotone under CP trace non-increasing map that preserves the trace of the first state \cite[Theorem 1']{muller2017monotonicity}, and  for $\alpha\in [1/2,1)$, we have by operator Jensen inequality for the power function, \[ D_\alpha(\rho\| e\sigma e )\le D_\alpha(e\rho e\| \sigma  ),\]
for $e$ being the support of $\rho$. The first inequality  follows from the monotonicity and $eV\rho V^*e=V\rho V^*$. The last inequality is due to $V\sigma V^*\leq \sigma\otimes {\bf 1}_E $
for $\sigma\in \N_{+} $.
  
Hence we have equality throughout, yielding
\begin{align*} D_\alpha(\rho\| \N)
= \inf_{\sigma\in \S(\Ha_A)} D_\alpha(V\rho V^* \| \sigma\otimes {\bf 1}_E)=-H_\al(E|A)_{V \rho V^*}. 
\end{align*}

(iii) For the hypothesis-testing relative entropy, write $\beta_\eps(\eta\|\tau)$ for the optimized type-II error, so that $D_H^\eps(\eta\|\tau)=-\log\beta_\eps(\eta\|\tau)$. For any $\sigma\in\S(\Ha_A)$ and any feasible test $q$ for $\rho$, the test $VqV^*$ is feasible for $V\rho V^*$ and has type-II error
\[
\tr(VqV^*(\sigma\otimes {\bf 1}_E))=\tr(q\E(\sigma)).
\]
Hence $\beta_\eps(V\rho V^*\|\sigma\otimes {\bf 1}_E)\leq \beta_\eps(\rho\|\E(\sigma))$. This gives
\[
\inf_{\sigma\in\S(\Ha_A)}D_H^\eps(V\rho V^*\|\sigma\otimes {\bf 1}_E)\geq D_H^\eps(\rho\|\E(\sigma))\ge  D_H^\eps(\rho\|\N).
\]
For the reverse direction, if $\tau\in\S(\N)$, then the same dilation of tests gives the other inequality. For any feasible test $Q$ for $V\rho V^*$, the test $V^*QV$ is feasible for $\rho$. Now using $V\tau V^*\leq \tau\otimes {\bf 1}_E$, we obtain 
\[
\tr(V^*QV\tau)=\tr(QV\tau V^*)\leq\tr(Q(\tau\otimes {\bf 1}_E)).
\]
This gives $\beta_\eps(\rho\|\tau)\leq \beta_\eps(V\rho V^*\|\tau\otimes {\bf 1}_E)$, for $\tau\in\S(\N)$. 
Taking logarithm and negative sign we obtain $D_H^\eps(\rho\|\tau)\geq D_H^\eps(V\rho V^*\|\tau\otimes {\bf 1}_E)$.
Taking the infimum over $\tau\in\S(\N)$ yields, \[D_H^\eps(\rho\|\N)\geq \inf_{\tau\in \S(\N)}D_H^\eps(V\rho V^*\|\tau\otimes {\bf 1}_E) \geq \inf_{\sigma\in\S(\Ha_A)}D_H^\eps(V\rho V^*\|\sigma\otimes {\bf 1}_E).\]
This finishes the proof.
\end{proof}

\begin{remark}{\rm \label{rem:high}
   On a high-level, the main arguments of the above proof rely on three aspects of the (smoothed) sandwiched R\'enyi entropies, namely: isometric invariance, data processing inequality under partial isometry and operator monotonicity of the second argument. It is likely that a general unified proof of the above result ought to be possible for any divergence satisfying these three properties. However, for our proof, we keep it separate as the smoothed min- and max-entropies require a bit more extra care. We are grateful to the referees for pointing out this fact. } 
\end{remark}
\begin{remark}{\rm 
For the special case of $\N$ being "the maximal abelian subalgebra" of $B(\Ha_A)$, the connection between the min- and the max- relative entropy of coherence with the conditional $H_{\max}$ and $H_{\min}$ entropies has been obtained by Zhao \emph{et al.} \cite{zhao-17}. Here we prove this connection between the $\eps$-smoothed max- and min- entropies and for sandwiched $\alpha$-R\'enyi relative entropies with the corresponding conditional entropies for any subalgebra. This generalizes the previously known results for $\al=\frac{1}{2},\infty$ and the non-smoothed version $\eps=0$.}
\end{remark}

Combined with duality of smoothed and sandwiched conditional entropies (see \cite[Section 6.3.1]{tomamichel-book} and \cite[Theorem 10]{muller2013quantum}), we have the following identities.
\begin{corollary}\label{cor: cond-entropy-max}
Let $|\psi\rangle \in B(\Ha_{A}\otimes \Ha_F)$ be a purification of $\rho$. Then $\xi_{EAF}=(V \otimes {\bf 1}_F) |\psi\rangle\langle \psi | (V\otimes {\bf 1}_F)^*$ is a purification of $V\rho V^*$, and for $\eps\in [0,1)$
\begin{align}
&D_{\max}^\eps(\rho\| \N)=- H_{\min}^\eps (E|A)_{\xi} =H_{\max}^\eps (E|F)_{\xi}\pl,\\
&D_{\min}^\eps(\rho\| \N)=- H_{\max}^\eps (E|A)_{\xi} =H_{\min}^\eps (E|F)_{\xi}\pl,\\
&D_{\al}(\rho\| \N)=- H_{\al} (E|A)_{\xi} =H_{\beta} (E|F)_{\xi}\pl,
\end{align}
for $2=\frac{1}{\al}+\frac{1}{\beta}, \al\in [\frac{1}{2},\infty]$.
\end{corollary}
It is interesting that the conditional entropies can be negative or positive but the subalgebra entropies are always non-negative.

Now we are ready to state and prove the following main theorem of this section.

 \begin{theorem}\label{max-AEP} For any subalgebra $\N\subseteq B(\Ha)$ and a state $\rho\in B(\Ha)$, and for any $\eps\in (0,1)$, we have
\begin{align*}&\lim_{n\to \infty}\frac{1}{n}D_{\max}^{\eps}(\rho^{\otimes n} \| \N^{\otimes n})
=\lim_{n\to \infty}\frac{1}{n}D_{\min}^{\eps}(\rho^{\otimes n} \| \N^{\otimes n})=\lim_{n\to \infty}\frac{1}{n}D_{H}^{\eps}(\rho^{\otimes n} \| \N^{\otimes n})=D(\rho\|\N).\end{align*}
\end{theorem}
\begin{proof}
It is known (see \cite[Corollary 6.24, 6.25]{tomamichel-book}) that 
\begin{align*}\lim_{n\to \infty}\frac{1}{n}H_{\min}^\eps (E^n|A^n)_{\omega^{\otimes n}}=\lim_{n\to \infty}\frac{1}{n}H_{\max}^\eps (E^n|A^n)_{\omega^{\otimes n}}=H(E|A)_\omega.
\end{align*}



Now from Theorem \ref{theorem: duality} we have the identity for  $\rho^{\ten n}$,
\[ D_{\max}^\eps(\rho^{\ten n}\| \N^{\ten n})=- H_{\min}^\eps (E^n|A^n)_{(V\rho V^*)^{\ten n}}
\]



Hence we have for any $\eps\in (0,1)$
\begin{align*}
\lim_{n\to \infty}\frac{1}{n}D_{\max}^{\eps}(\rho^{\otimes n} \| \N^{\otimes n})= \lim_{n\to \infty}-\frac{1}{n} H_{\min}^{\eps} (E^{ n} |A^{n})_{(V\rho V^*)^{\otimes n}}= -H(E|A)_{(V\rho V^*)}= D(\rho\|\N),
\end{align*}
where the last equality is again from the duality Theorem \ref{theorem: duality} for standard relative entropy $D$ ($\al=1$ in the $D_\al$ case). Hence we get \begin{equation}\label{eq-max-AEP}
    \lim_{n\to \infty}\frac{1}{n}D_{\max}^{\eps}(\rho^{\otimes n} \| \N^{\otimes n})=D(\rho\|\N).
\end{equation}

\noindent The argument for $D_{\min}^\eps$ is similar (via $H_{\max}^{\eps}$). For hypothesis testing relative entropy, we use the following relation (see \cite{anshu}) between $D_{\max}^{\eps}$ and $D_{H}^{\eps}$, for $\eps\in (0,1)$ , 
\begin{align}D_{\max}^{\sqrt{1-\eps}}(\rho\|\sigma)\leq D_{H}^{\eps}(\rho\|\sigma)+\log \frac{1}{\eps}\pl. 
\label{eq:3}
\end{align}
Now taking minimum over all states in $\N$ we have
\[D_{\max}^{\sqrt{1-\eps}}(\rho\|\N)\leq D_{H}^{\eps}(\rho\|\N)+\log \frac{1}{\eps}.\]
Now, for any $\eps\in(0,1)$, using equation (\ref{eq-max-AEP}) we have 
\begin{align*}&\liminf_{n\to \infty}\frac{1}{n} D_{H}^{\eps}(\rho^{\otimes n}\|\N^{\otimes n})
\geq \lim_{n\to \infty}\frac{1}{n}   (D_{\max}^{\sqrt{1-\eps}}(\rho^{\otimes n}\|\N^{\otimes n})-\log \frac{1}{\eps})
= D(\rho\|\N).
\end{align*}
On the other hand,
\begin{align*}\limsup_{n\to\infty} \frac{1}{n} D_{H}^{\eps}(\rho^{\otimes n}\|\N^{\otimes n})\leq \lim_{n\to\infty}\frac{1}{n} D_{H}^{\eps}(\rho^{\otimes n}\|\E(\rho)^{\otimes n})=D(\rho\|\E(\rho))=D(\rho\|\N),
\end{align*}
where the inequality above follows from the definition of $D_{H}^{\eps}(\rho^{\otimes n}\|\N^{\otimes n})$ and the  next equality follows from the Stein's lemma of two states (\cite{HP, ON}).
That completes the proof.
\end{proof}


\subsection{Second-order asymptotics}\label{subsec:second order}
One advantage of our dilation theorem is that it implies the second-order asymptotics, which is beyond the reach of generalized Stein's lemma for general convex sets of states.  

For a state $\rho$ and a positive semidefinite operator $\sigma$, the following second-order asymptotics was proved by Hayashi and Tomamichel \cite{tomamichel2013hierarchy} (see also \cite{Ke-Li}),
\begin{align}\label{eq:dmax2} &D^\eps_{\max}(\rho^{\ten n}||\sigma^{\ten n})=n D(\rho||\sigma)-\sqrt{n}\sqrt{V(\rho||\sigma)}\Phi^{-1}(\eps^2)+O(\log n)\pl, \\
&D^\eps_{H}(\rho^{\ten n}||\sigma^{\ten n})\nonumber=n D(\rho||\sigma)+\sqrt{n}\sqrt{V(\rho||\sigma)}\Phi^{-1}(\eps)+O(\log n).\nonumber
\end{align}
Here
\[V(\rho||\sigma):=\tr(\rho(\log \rho-\log\sigma)^2)-D(\rho||\sigma)^2\]
is the quantum relative variance and $\Phi^{-1}$ is the inverse function of the cumulative distribution function of a standard normal random variable $\Phi$ defined as follows
\begin{align*}& \Phi(a)=\frac{1}{\sqrt{2\pi}}\int_{-\infty}^a e^{-\frac{x^2}{2}}dx\pl,\pl \Phi^{-1}(\eps)=\sup\{ a\in \mathbb{R} | \Phi(a)\le \eps\}\pl. 
\end{align*}
The term $O(\log n)$ hides constants that depend on $\rho, \sigma$ and $\eps$. 
Note that in this paper we use the purified distance $d(\rho,\rho')\le \eps$ for the smoothing parameter $\eps$, following \cite{tomamichel2013hierarchy}. The choice of the smoothing criterion does not affect the first order but matters for the second order. Nuradha and Wilde \cite{nuradha2024fidelity} proved the second-order asymptotics for $D_{\min}^\eps$,  
\begin{align}\label{eq:dmin2}D^\eps_{\min}(\rho^{\ten n}||\sigma^{\ten n})
=n D(\rho||\sigma)+\sqrt{n}\sqrt{V(\rho||\sigma)}\Phi^{-1}(\eps^2)+O(\log n)\pl. \end{align}
The equations \eqref{eq:dmax2} and \eqref{eq:dmin2} immediately imply the second-order asymptotics for the alternative smoothing condition for the min-entropy and the max-entropy without optimization over $B$,
\begin{align*}
&\tilde{H}_{\min}(A|B)_{\rho}=-D_{\max}(\rho_{AB}||{\bf 1}_A\ten \rho_B)\pl,\pl  \tilde{H}_{\min}^\eps(A|B)_{\rho}=-D_{\max}^\eps(\rho_{AB}|| {\bf 1}_A\ten \rho_B)\\
&\tilde{H}_{\max}(A|B)_{\rho}=-D_{\min}(\rho_{AB}|| {\bf 1}_A\ten \rho_B)\pl ,\pl  \tilde{H}_{\max}^\eps(A|B)_{\rho}=-D_{\min}^\eps(\rho_{AB}||{\bf 1}_A\ten \rho_B), \\
&\tilde{H}^\eps_{h}(A|B)_{\rho}=-D_{H}^{\eps}(\rho_{AB}|| {\bf 1}_A\ten \rho_B)\pl ,
\end{align*}
Using duality relations, one can further derive the second-order expansion for the smooth conditional min entropy with optimization.
\begin{prop}\label{prop:duality1}For a bipartite state $\rho_{AB}\in \mathcal{B}(\mathcal{H}_A\ten \mathcal{H}_B)$,
\begin{align*} H^\eps_{\min}(A^n|B^n)_{\rho^{\otimes n}}
&=n H(A|B)_\rho+
\sqrt{n}\sqrt{V(\rho_{AB}||{\bf 1}_A\ten \rho_B)}\Phi^{-1}(\eps^2)+O(\log n),\\
H^\eps_{\max}(A^n|B^n)_{\rho^{\otimes n}}
&=n H(A|B)_\rho-\sqrt{n}\sqrt{V(\rho_{AB}|| {\bf 1}_A\ten \rho_B)}\Phi^{-1}(\eps^2)+O(\log n)\pl, \\
H^\eps_{h}(A^n|B^n)_{\rho^{\otimes n}}&=n H(A|B)_\rho-\sqrt{n}\sqrt{V(\rho_{AB}|| {\bf 1}_A\ten \rho_B)}\Phi^{-1}(\eps)+O(\log n)\pl.
\end{align*}
\end{prop}
The proof of this proposition uses the duality of var-entropy in the following lemma. Note also that the lemma below is a special case of Eq. (3.15) in \cite{HT}.
\begin{lemma}
For a pure state $\rho_{ABC}\in \mathcal{B}(\mathcal{H}_A\ten \mathcal{H}_B\ten \mathcal{H}_C)$,
\[ V(\rho_{AB}|| {\bf 1}\ten \rho_B)=V(\rho_{AC}|| {\bf 1}\ten \rho_C)\pl. \]
\end{lemma}
\begin{proof}
By definition,
\begin{align*}
V(\rho_{AB}||{\bf 1}_A\ten \rho_B)&=\tr(\rho_{AB}(\log \rho_{AB}
-\log {\bf 1}_A\ten \rho_B)^2)-D(\rho_{AB}||{\bf 1}_A\ten \rho_B)^2
\end{align*}
For the second term, by the duality $H(A|B)_\rho=-H(A|C)_\rho$ for pure state
\begin{align*} &D(\rho_{AB}||{\bf 1}_A\ten \rho_B)^2=H(A|B)_\rho^2=H(A|C)_\rho^2=D(\rho_{AC}||{\bf 1}_A\ten \rho_C)^2 \end{align*}
For the first term, we use the fact for a bipartite pure state $\ket{\phi_{AB}}$, for any function $f$,
\[ f(\phi_A)\ket{\phi_{AB}}=f(\phi_B)\ket{\phi_{AB}}\pl. \]
Indeed, assume the Schmidt decomposition $\ket{\phi_{AB}}= \sum_{i}\sqrt{\la_{i}}\ket{i_A}\ket{i_B}$. Then $\phi_A=\sum_{i}\la_i\ketbra{i_A}, \phi_B=\sum_{i}\la_i\ketbra{i_B}$
\[ f(\phi_A)\ket{\phi_{AB}}=\sum_{i} f(\la_i)\sqrt{\la_i}\ket{i_A}\ket{i_B}=f(\phi_B)\ket{\phi_{AB}}  .\]
Applying this property to the pure state $\rho_{ABC}$, we have
\begin{align*} \log\rho_{AB}\ket{\rho_{ABC}}=\log\rho_{C}\ket{\rho_{ABC}}\pl, \log\rho_{B}\ket{\rho_{ABC}}=\log\rho_{AC}\ket{\rho_{ABC}}
\end{align*}
Hence
\begin{align*}
&\tr(\rho_{AB}(\log \rho_{AB}-\log {\bf 1}_A\ten \rho_B)^2)
\\ =& \bra{\rho_{ABC}}(\log \rho_{AB}-\log({\bf 1}_A\ten \rho_B))^2\ket{\rho_{ABC}}
\\ =& \bra{\rho_{ABC}}(\log \rho_{AB})^2+(\log({\bf 1}_A\ten \rho_B))^2-\log \rho_{AB}\cdot \log({\bf 1}_A\ten \rho_B)-\log({\bf 1}_A\ten \rho_B)\cdot \log \rho_{AB}\ket{\rho_{ABC}}
\\ =& \bra{\rho_{ABC}}(\log \rho_{C})^2+(\log \rho_{AC})^2-\log \rho_{C}\cdot \log  \rho_{AC}-\log \rho_{AC}\cdot \log \rho_{C}\ket{\rho_{ABC}}
\\ =& \bra{\rho_{ABC}}(\log \rho_{C}-\log \rho_{AC})^2 \ket{\rho_{ABC}}
\\ =& \tr(\rho_{AC}(\log \rho_{AC}-\log  \rho_C)^2)
\end{align*}
That finishes the proof.
\end{proof}

\begin{proof}[Proof of Proposition \ref{prop:duality1}]
By definition 
\begin{align*}&\tilde{H}^\eps_{\min}(A|B)_{\rho}\le {H}_{\min}^\eps(A|B)_{\rho}\pl, 
\tilde{H}_{\max}^\eps(A|B)_{\rho}\le {H}_{\max}^\eps(A|B)_{\rho}\pl,\\
&\tilde{H}_{h}^\eps(A|B)_{\rho}\le {H}_{h}^\eps(A|B)_{\rho}.
\end{align*}
On one hand, using the second-order asymptotics of $\tilde{H}_{\min}(A|B)$, we have
\begin{align*}
{H}_{\min}^\eps(A^n|B^n)_{\rho^{\otimes n}}\ge &\tilde{H}_{\min}^\eps(A^n|B^n)_{\rho^{\otimes n}} =n H(A|B)_\rho + \sqrt{n}\sqrt{V(\rho_{AB}||{\bf 1}\ten \rho_B)}\Phi^{-1}(\eps^2)+\mathcal{O}(\log n).
\end{align*}
On the other hand, by duality over the purification $\rho_{ABC}$
 \begin{align*} H(A|B)_\rho=-H(A|C)_\rho, V(\rho_{AB}||{\bf 1}\ten \rho_B)=V(\rho_{AC}||{\bf 1}\ten \rho_C)\pl. 
\end{align*}
 So we have
 \begin{align*}
{H}_{\min}^\eps(A^n|B^n)_{\rho^{\otimes n}}=&-{H}_{\max}^\eps(A^n|C^n)_{\rho^{\otimes n}}
\le -\tilde{H}_{\max}^\eps(A^n|C^n)_{\rho^{\otimes n}}\\
 =&-n H(A|C)_\rho+ \sqrt{n}\sqrt{V(\rho_{AC}||{\bf 1}\ten \rho_C)}\Phi^{-1}(\eps^2)+\mathcal{O}(\log n)
\\ =&n H(A|B)_\rho+ \sqrt{n}\sqrt{V(\rho_{AB}||{\bf 1}\ten \rho_B)}\Phi^{-1}(\eps^2)+\mathcal{O}(\log n).
 \end{align*}
 This proves the asymptotics for ${H}_{\min}^\eps(A^n|B^n)_{\rho^{\otimes n}}$ and by duality also for ${H}_{\max}^\eps(A^n|B^n)_{\rho^{\otimes n}}$. For hypothesis testing conditional entropy, on one hand, we have, from the second-order expansion for two states,
 \begin{align*}
 &{H}_{h}^\eps(A^n|B^n)_{\rho^{\otimes n}}
 \ge \tilde{{H}}_{h}^\eps(A^n|B^n)_{\rho^{\otimes n}}  =n H(A|B)_\rho- \sqrt{n}\sqrt{V(\rho_{AB}||{\bf 1}\ten \rho_B)}\Phi^{-1}(\eps)
 +\mathcal{O}(\log n).
 \end{align*}
On the other hand, by the inequality \eqref{eq:3} \begin{align*}&{H}_{h}^\eps(A^n|B^n)_{\rho^{\otimes n}}
\le  {H}_{\min}^{\sqrt{1-\eps}}(A^n|B^n)_{\rho^{\otimes n}}+\log \frac{1}{\varepsilon} \\ =&n H(A|B)_\rho- \sqrt{n}\sqrt{V(\rho_{AB}||{\bf 1}\ten \rho_B)}\Phi^{-1}(\eps)
+\mathcal{O}(\log n) \pl.  \end{align*}
That completes the proof.
\end{proof}

\begin{theorem}\label{max-AEP-second order}Let $\N\subseteq B(\Ha_A)$ be a subalgebra and let $V: \Ha_A\rightarrow \Ha_A\otimes \Ha_E$ be a Stinespring isometry for the trace-preserving conditional expectation $\E: B(\Ha_A)\rightarrow \N$ such that
$\E(x)=V^*(x\otimes {\bf 1}_E)V$. Then for any  state $\rho\in B(\Ha_A)$ and for any $\eps\in (0,1)$, we have
\begin{align*}
D_{\max}^\eps (\rho^{\ten n}||\N^{\ten n})=&nD(\rho||\N)-\sqrt{n}\sqrt{V(\rho||\E(\rho))}\Phi^{-1}(\eps^2)
+\mathcal{O}(\log n)\pl, \\
D_{\min}^\eps (\rho^{\ten n}||\N^{\ten n})=&n D(\rho||\N)+\sqrt{n}\sqrt{V(\rho||\E(\rho))}\Phi^{-1}(\eps^2)+\mathcal{O}(\log n)\pl, \\
D_{H}^\eps (\rho^{\ten n}||\N^{\ten n})=&n D(\rho||\N)+\sqrt{n}\sqrt{V(\rho||\E(\rho))}\Phi^{-1}(\eps)+\mathcal{O}(\log n)\pl.
\end{align*}
\end{theorem}
\begin{proof}
Using the dilation Theorem \ref{theorem: duality} for $n$-copy, 
\begin{align*} &D_{\max}^\eps (\rho^{\otimes n}||\N^{\otimes n})=-H_{\min}^\eps(E^n|A^n)_{(V\rho V^*)^{\otimes n}}\pl ,D_{\min}^\eps (\rho^{\otimes n}||\N^{\otimes n})=-H_{\max}^\eps(E^n|A^n)_{(V\rho V^*)^{\otimes n}}\pl ,\pl \\
& D_{\alpha} (\rho^{\otimes n}||\N^{\otimes n})=-H_{\alpha}(E^n|A^n)_{(V\rho V^*)^{\otimes n}}\pl,\pl D_{H}^\eps (\rho^{\otimes n}||\N^{\otimes n})=-H_{h}^\eps(E^n|A^n)_{(V\rho V^*)^{\otimes n}}.
\end{align*}
It suffices to justify the second-order term that for any state $\rho\in \mathcal{B}(\Ha_A)$
\[ V(\rho||\E(\rho))= V(V\rho V^*|| \E(\rho) \ten {\bf 1}_E) \pl.\]
On one hand, it was proved in \cite[Proposition 11]{HT} that for any bipartite state $\sigma_{AB}$,
\[\frac{1}{2}V(\sigma_{AB}|| {\bf 1}\ten \sigma_B)=-\frac{d}{d\alpha }H_\al(A|B)_\sigma|_{\alpha=1}\pl.\]
Then by the dilation of sandwiched $\alpha$-R\'enyi relative entropy $D_\alpha$ we have,
\begin{align*}
&\frac{d}{d\alpha }D_\alpha(\rho||\N )|_{\alpha=1}=-\frac{d}{d\alpha }H_\al(E|A)_{V\rho V^*}|_{\alpha=1}=\frac{1}{2}V(V\rho V^*||\E(\rho) \ten {\bf 1}_E) \pl,
\end{align*}
where $\E(\rho)=\tr_{E}(V\rho V^*)$ is the reduced density of $V\rho V^*$ on $A$, and we have used the fact that $\E=\E^{\dagger}$ (see equation \eqref{eq-selfdual}). 
On the other hand, for $\alpha\ge \frac{1}{2}$, denote \[\sigma_\alpha=\arg\min_{\sigma\in\S(\N)} D_\alpha(\rho||\sigma )\pl.\]
Then
\begin{align*} \frac{d}{d\alpha }D_\alpha(\rho||\N )|_{\alpha=1}
=&\frac{d}{d\alpha }\Big(\min_{\sigma\in\S(\N) }D_\alpha(\rho||\sigma )\Big)|_{\alpha=1}
= \frac{d}{d\alpha }D_\alpha(\rho||\sigma_\alpha )|_{\alpha=1}\\
\overset{(1)}{=} &\frac{d}{d\alpha }D_\alpha(\rho||\E(\rho) )|_{\alpha=1}
\\ \overset{(2)}{=}&\frac{1}{2}V(\rho||\E(\rho))
\end{align*}
Here (1) uses \cite[Lemma 24]{HT} and the fact $D_\al(\rho||\sigma)$ is strictly convex with respect to $\sigma$ for all $1/2<\alpha<\infty$, and (2) uses \cite[Theorem 2]{lin2015investigating}.
This completes the proof.
\end{proof}

\begin{remark}{\rm By the dilation argument above, we also obtain the second-order expansion of $\alpha\mapsto D_\al(\rho||\N)$,
\begin{align*}D_\al(\rho||\N)=D(\rho||\N)+\frac{(\al-1)}{2}V(\rho||\E(\rho))+\mathcal{O}(|\al-1|^2)\pl.\end{align*} }
\end{remark}

\section{Resource theory of subalgebra coherence}\label{resource theory}

\subsection{Generalized coherence}
The Stein's lemma and AEP above suggest that one can consider a resource theory whose free states are the states contained in a fixed subalgebra. Let $\Ha$ be a Hilbert space of dimension $d$ and let $\N\subseteq B(\Ha)$ be a von Neumann subalgebra. Before introducing the generalized coherence, we recall the Pimsner-Popa index of a von Neumann subalgebra \cite{pimsner-popa}, which was motivated by Jones index theory for subfactors \cite{Jones}.

\begin{definition}\label{Popa-pimsner}
Let $\E_\N:B(\Ha)\to \N$ be the trace-preserving conditional expectation. We define
\[
\lambda_\N:=\lambda(B(\Ha):\N)=\max\{\lambda>0: \lambda x\leq \E_\N(x),\ \forall x\in B(\Ha)_+\}.
\]
Here $B(\Ha)_+$ denotes the cone of positive semidefinite operators on $\Ha$.
\end{definition}

This index was introduced mainly in subfactor theory, but here we restrict ourselves to the finite-dimensional setting. By the structure theorem for finite-dimensional unital $^*$-algebras, after a unitary conjugation we may write
\[
\N=\bigoplus_k \bigl(\M_{m_k}(\C)\otimes 1_{n_k}\bigr),
\qquad \text{dim}(\mathcal{H})=\sum_k m_kn_k=:d .
\]
With this convention, \cite[Theorem 6.1]{pimsner-popa} gives
\[
\lambda_\N^{-1}=\sum_k \min\{m_k,n_k\}\, n_k .
\]
In particular, $\lambda_\N\leq 1$, with equality if and only if $\N=B(\Ha)$. For the diagonal subalgebra, $\lambda_\N=1/d$; if $B(\Ha)=\M_m(\C)\otimes\M_n(\C)$ and $\N=\M_m(\C)\otimes 1_n$, then $\lambda_\N=1/(\min\{m,n\}n)$.

We say that a quantum state $\rho\in S(\Ha)$ is \textbf{incoherent with respect to $\N$}, or simply \textbf{$\N$-incoherent}, if $\rho\in\S(\N):=\N\cap S(\Ha)$. The subalgebra entropies discussed in the previous sections can be viewed as $\N$-coherence measures. Indeed, for $\mathbb D=D_{\max},D_{\min}$ and $D_\alpha$ with $\alpha\in(1/2,\infty)$, one has $\mathbb D(\rho\|\N)=0$ if and only if $\rho\in \S(\N)$. It is proved in \cite{GJL19} that these $\N$-coherence measures have the same maximum:
\[
\max_{\rho\in S(\Ha)}D_\alpha(\rho\|\N)
=\max_{\rho\in S(\Ha)}D_\alpha(\rho\|\E_\N(\rho))
=\log\lambda_\N^{-1} .
\]
The key observation, already implicit in Pimsner and Popa's work, is that for any propert finite-dimensional inclusion $\N\subset B(\Ha)$ there exists a projection $p_\N$ such that
\[
\E_\N(p_\N)=\lambda_\N f_\N,
\]
where $f_\N\in\N$ is another projection. The normalized flat state
\[
\omega_\N:=\frac{p_\N}{\tr(p_\N)}
\]
then $\omega_\N$ attains the maximal value, that is,  $D_\alpha(\omega_\N\|\N)=\log\lambda_\N^{-1}$ for all $\alpha\in[1/2,\infty]$. We call such a state a \textbf{maximally $\N$-coherent state}.

For example, when $\N=\operatorname{span}\{|j\rangle\langle j|:j=1,\ldots,d\}$ is the diagonal subalgebra, one may take $\omega_\N=|\phi\rangle\langle\phi|$ with $|\phi\rangle=d^{-1/2}\sum_{i=1}^d|i\rangle$, and the dephasing map gives
\[
\E_\N(|\phi\rangle\langle\phi|)=\sum_j |\langle j|\phi\rangle|^2 |j\rangle\langle j|=\frac1d 1 .
\]
In general, $\E_\N(\omega_\N)$ need only be proportional to a projection, not to the identity. For instance, if $B(\Ha)=\M_m(\C)\otimes\M_n(\C)$, $\N=\M_m(\C)\otimes 1_n$, and $m>n$, then one may take $\omega_\N=|\psi\rangle\langle\psi|$ with $|\psi\rangle=n^{-1/2}\sum_{i=1}^n |i\rangle|i\rangle$. The trace-preserving conditional expectation onto $\N$ is then the partial trace map
\[
\E_\N(\omega_\N)= \frac{1}{n}\id\otimes \tr_{n}(\omega_\N)\otimes 1_n=\frac{1}{n^2}\sum_{i=1}^n |i\rangle\langle i|\otimes 1_n,
\]
which is flat on its support but not full rank when $m>n$.

\subsection{Incoherent operations}
We consider quantum systems $\Ha$ equipped with distinguished subalgebras $\N\subset B(\Ha)$. For two such systems $\N\subset B(\Ha)$ and $\M\subset B(\Ka)$, let $\E_\N$ and $\E_\M$ denote the corresponding trace-preserving conditional expectations.

\begin{definition}
Let $\Phi:B(\Ha)\to B(\Ka)$ be a quantum channel.
\begin{enumerate}
\item[i)] We say that $\Phi$ is an \textbf{$\N$-$\M$ maximally incoherent operation (MIO)} if $\Phi(\S(\N))\subseteq\S(\M)$. Equivalently, $\Phi$ is MIO if
\[
\Phi\circ\E_\N=\E_\M\circ\Phi\circ\E_\N .
\]
\item[ii)] We say that $\Phi$ is an \textbf{$\N$-$\M$ dephasing-covariant incoherent operation (DIO)} if
\[
\Phi\circ\E_\N=\E_\M\circ\Phi .
\]
\end{enumerate}
When the subalgebras are clear from context, we simply call such maps MIO or DIO operations.
\end{definition}

An immediate property of MIO operations is monotonicity:
\[
D_\alpha(\rho\|\N)\geq D_\alpha(\Phi(\rho)\|\M),\qquad \alpha\in[1/2,\infty].
\]
This follows from the data-processing inequality for $D_\alpha$ and  the inclusion $\Phi(\S(\N))\subseteq\S(\M)$. Every DIO operation is MIO. For example, the conditional expectation $\E_\N:B(\Ha)\to\N$ is the analogue of the completely dephasing map and is both $\N$-$\N$ MIO and $\N$-$\N$ DIO. We do not discuss analogues of incoherent operations (IO) or strictly incoherent operations (SIO), although such notions can also be defined in the present subalgebra setting.

\subsection{Dilution with respect to subalgebra coherence}
For each inclusion $\N\subset B(\Ha)$, fix a maximally $\N$-coherent state $\omega_\N=p_\N/\tr(p_\N)$ as above. We now introduce the subalgebra coherence cost, in analogy with the coherence cost in \cite{dilution}. In this subsection, the smoothing condition is the fidelity condition $F(\rho,\rho')\geq 1-\eps$.

\begin{definition}
Let $\rho\in S(\Ha)$ and let $\N\subseteq B(\Ha)$ be a subalgebra. Given $\eps\in[0,1]$, define the \textbf{one-shot $\eps$-$\N$-coherence cost} under operations $\mathcal O$ by
\[
C_{\N,\mathcal O}^\eps(\rho)
:=\inf_{(\M,\omega_\M)}\inf_{\Phi\in\mathcal O}
\left\{\log\lambda_\M^{-1}: F(\rho,\Phi(\omega_\M))\geq 1-\eps\right\},
\]
where $\lambda_\M=\lambda(B(\Ka):\M)$, the first infimum ranges over all pairs $(\M, \omega_\M)$ where $\M\subseteq B(\Ka)$ is a subalgebra and $\omega_\M$ a maximally $\M$-coherent state, and the second infimum ranges over all $\M$-$\N$ incoherent quantum operations $\Phi:B(\Ka)\to B(\Ha)$ in the class $\mathcal O$.
The asymptotic $\N$-coherence cost is
\[
C_{\N,\mathcal O}^\infty(\rho)
=\lim_{\eps\to0}\limsup_{n\to\infty}\frac1n C_{\N,\mathcal O}^\eps(\rho^{\otimes n}).
\]
\end{definition}

The definition depends on the chosen maximally $\M$-coherent state, but the following theorem shows that this choice does not affect the asymptotic MIO cost.

\begin{theorem}\label{thm 1: dilution}
Let $\rho\in S(\Ha)$ and let $\N\subseteq B(\Ha)$ be a subalgebra. For every $\eps\in[0,1]$,
\[
D_{\max}^\eps(\rho\|\N)\leq C_{\N,MIO}^\eps(\rho)\leq D_{\max}^\eps(\rho\|\N)+1,
\]
where the smoothing in $D_{\max}^\eps$ is by the fidelity condition $F(\rho,\rho')\geq1-\eps$. Consequently,
\[
C_{\N,MIO}^\infty(\rho)=D(\rho\|\N).
\]
\end{theorem}

\begin{proof}
The proof is similar to the case of coherence \cite{dilution}. For the converse bound, fix $\delta>0$ and choose a feasible triple $(\M,\omega_\M,\Phi)$ such that
\[
\log\lambda_\M^{-1}\leq C_{\N,MIO}^\eps(\rho)+\delta,
\qquad F(\rho,\Phi(\omega_\M))\geq1-\eps .
\]
Set $\rho'=\Phi(\omega_\M)$. Then
\begin{align*}
D_{\max}^\eps(\rho\|\N)
&\leq D_{\max}(\rho'\|\N)\\
&\leq D_{\max}(\Phi(\omega_\M)\|\Phi(\E_\M(\omega_\M)))\\
&\leq D_{\max}(\omega_\M\|\E_\M(\omega_\M))\\
&=\log\lambda_\M^{-1}
\leq C_{\N,MIO}^\eps(\rho)+\delta .
\end{align*}
Here the second line uses that $\Phi(\E_\M(\omega_\M))\in\S(\N)$, and the third line uses data processing. Letting $\delta\downarrow0$ gives the lower bound.

For achievability, choose $\rho'$ and $\sigma\in\S(\N)$ such that $F(\rho,\rho')\geq1-\eps$ and
\[
D_{\max}^\eps(\rho\|\N)=D_{\max}(\rho'\|\sigma)=\log\mu .
\]
If $\mu=1$, then $\rho'=\sigma\in\S(\N)$ and the cost is zero. Thus assume $\mu>1$ and set $n=\lceil\mu\rceil\ge2$. Let $\K=\C^n$, let $\M\subset B(\K)$ be the diagonal subalgebra, and let
\[
\omega_\M=|\phi_n\rangle\langle\phi_n|,
\qquad |\phi_n\rangle=\frac1{\sqrt n}\sum_{i=1}^n |i\rangle,
\]
so that $\E_\M(\omega_\M)=1_n/n$. Define
\[
\Phi(x)=\frac{n}{n-1}\bigl(\tr(x)-\tr(\omega_\M x)\bigr)\left(\sigma-\frac1n\rho'\right)+\tr(\omega_\M x)\rho'.
\]
Since $\rho'\leq\mu\sigma\leq n\sigma$, the operator $\sigma-\rho'/n$ is positive; hence $\Phi$ is completely positive. A direct trace calculation shows that $\Phi$ is trace-preserving. If $\gamma\in\S(\M)$, then $\tr(\omega_\M\gamma)=1/n$, and therefore $\Phi(\gamma)=\sigma\in\S(\N)$. Thus $\Phi$ is an $\M$-$\N$ MIO operation. Also $\Phi(\omega_\M)=\rho'$. Hence
\[
C_{\N,MIO}^\eps(\rho)\leq\log n\leq \log(1+\mu)\leq 1+\log\mu
=1+D_{\max}^\eps(\rho\|\N).
\]
The asymptotic statement follows from the AEP for $D_{\max}^\eps$ after translating between the fidelity and purified-distance smoothing parameters.
\end{proof}

\begin{remark}{\rm 
The strong AEP gives a strong-converse form: for each fixed fidelity error $\eps\in(0,1)$,
\[
\lim_{n\to\infty}\frac1n C_{\N,MIO}^\eps(\rho^{\otimes n})=D(\rho\|\N).
\]
In the achievability proof above we used the standard diagonal coherence resource to dilute an arbitrary $\N$-coherent state. One can use any maximally coherent state $\omega_\M$ with respect to a chosen subalgebra.}
\end{remark}

For DIO operations, define
\[
D_{\max,\E_\N}(\rho):=D_{\max}(\rho\|\E_\N(\rho))
=\log\min\{\lambda>0:\rho\leq\lambda\E_\N(\rho)\},
\]
and, with fidelity smoothing,
\[
D_{\max,\E_\N}^\eps(\rho):=
\inf_{\rho':F(\rho,\rho')\ge1-\eps}D_{\max,\E_\N}(\rho').
\]
Clearly $D_{\max,\E_\N}(\rho)\geq D_{\max}(\rho\|\N)$ and $D_{\max,\E_\N}^\eps(\rho)\geq D_{\max}^\eps(\rho\|\N)$.

\begin{theorem}\label{thm:DIO-dilution}
Let $\rho\in S(\Ha)$ and let $\N\subseteq B(\Ha)$ be a subalgebra with trace-preserving conditional expectation $\E_\N$. For every $\eps\in[0,1]$,
\[
D_{\max,\E_\N}^\eps(\rho)\leq C_{\N,DIO}^\eps(\rho)
\leq D_{\max,\E_\N}^\eps(\rho)+1.
\]
\end{theorem}

\begin{proof}
The converse is analogous to the proof of Theorem \ref{thm 1: dilution}. Fix $\delta>0$ and choose a feasible DIO protocol with
$\log\lambda_\M^{-1}\le C_{\N,DIO}^\eps(\rho)+\delta$. Let $\rho'=\Phi(\omega_\M)$. Then
\begin{align*}
D_{\max,\E_\N}^\eps(\rho)
&\leq D_{\max,\E_\N}(\rho')\\
&=D_{\max}(\Phi(\omega_\M)\|\E_\N(\Phi(\omega_\M)))\\
&=D_{\max}(\Phi(\omega_\M)\|\Phi(\E_\M(\omega_\M)))\\
&\leq D_{\max}(\omega_\M\|\E_\M(\omega_\M))\\
=&\log\lambda_\M^{-1}\\
&\le C_{\N,DIO}^\eps(\rho)+\delta .
\end{align*}
Letting $\delta\downarrow0$ gives the lower bound.

For the upper bound, choose $\rho'$ such that $F(\rho,\rho')\ge1-\eps$ and
\[
D_{\max,\E_\N}^\eps(\rho)=D_{\max,\E_\N}(\rho')=\log\lambda,
\qquad \rho'\leq\lambda\E_\N(\rho').
\]
If $\lambda=1$, the cost is zero. Otherwise set $n=\lceil\lambda\rceil\ge2$, take the diagonal subalgebra $\M\subset B(\C^n)$, and let $\omega_\M=|\phi_n\rangle\langle\phi_n|$. Define
\[
\Phi(x)=\frac{n}{n-1}\bigl(\tr(x)-\tr(\omega_\M x)\bigr)
\left(\E_\N(\rho')-\frac1n\rho'\right)+\tr(\omega_\M x)\rho'.
\]
Since $\rho'\leq\lambda\E_\N(\rho')\leq n\E_\N(\rho')$, the map is completely positive, and it is trace-preserving by direct calculation. For every state $\gamma\in S(\C^n)$, $\tr(\omega_\M\E_\M(\gamma))=1/n$, whence
\[
\Phi(\E_\M(\gamma))=\E_\N(\rho').
\]
On the other hand, for every state $x\in S(\C^n)$,
\[
\E_\N(\Phi(x))=\E_\N(\rho').
\]
By linearity, $\E_\N\circ\Phi=\Phi\circ\E_\M$, so $\Phi$ is an $\M$-$\N$ DIO operation. Since $\Phi(\omega_\M)=\rho'$, we get
\[
C_{\N,DIO}^\eps(\rho)\leq\log n\leq1+\log\lambda
=1+D_{\max,\E_\N}(\rho')=1+D_{\max,\E_\N}^\eps(\rho).
\]
This concludes the proof.
\end{proof}

\section{Outlook}
In this paper we establish the generalized Stein's lemma for the state space of subalgebras. This motivated us to look deeper into the resource theory of coherence in this general setup. We  provided an operational meaning to the subalgebra entropy in the context of resource dilution. Our result can be used in the analyses of quantum information processing tasks that exploit coherence as resources, such as quantum key distribution and random number generation. For future directions, we would like to focus on the coherence distillation in particular, and establish the asymptotic coherence distillation results. Also, our approach allows us to extend the framework of the current work in the paradigm of infinite-dimensional setting, which we leave as a future endeavor. Indeed, since the index theory for subalgebras has been a widely popular topic in operator algebras, we expect to extend some of our results in a more general context.

\end{document}